\documentclass{article}

\usepackage{aaai20}
\usepackage{times}  
\usepackage{helvet}  
\usepackage{courier}  
\usepackage{url}  
\usepackage{graphicx}  
\usepackage{tikz}
\usepackage{pgflibraryshapes}  
\usetikzlibrary{positioning}  
\usetikzlibrary{patterns,arrows,shapes}
\usepackage{mathrsfs}
\usepackage{paralist}
\usepackage{tabularx}
\usepackage{diagbox}
\usepackage{nicefrac}
\usepackage[centertags,fleqn]{amsmath}
\usepackage{amssymb,amsfonts}
\usepackage{txfonts}
\usepackage{array,xspace,multirow,epic}
\usepackage{algorithm}
\usepackage[noend]{algorithmic}
\usepackage{eqparbox}
\usepackage{booktabs}  
\usepackage{verbatim,ifthen}
\usepackage{booktabs}  
\usepackage{enumitem}
\usepackage{multicol}
\usepackage{hyperref}

\usepackage{caption}
\usepackage{subcaption}

\usepackage{array,xspace,multirow,epic}
\usepackage{boxedminipage}
\usepackage{xspace}
\newcommand{\pbDef}[3]{%
\noindent
\begin{center}
\begin{boxedminipage}{0.98 \columnwidth}
#1\\[5pt]
\begin{tabular}{l p{0.75 \columnwidth}}
Input: & #2\\
Question: & #3
\end{tabular}
\end{boxedminipage}
\end{center}
}
\usepackage{etoolbox}
\usepackage{stfloats}

\newcounter{Bew1}
\newcounter{Bew2}

\usepackage{theorem}

\newtheorem{theorem}{Theorem}
\newtheorem{lemma}{Lemma}

\newtheorem{proposition}{Proposition}

\theoremstyle{definition}

\newtheorem{example}[theorem]{Example}
\newtheorem{remark}[theorem]{Remark}

\newcommand{\qed}{\unskip\hspace*{1em}\hspace{\fill}$\Box$}
\newenvironment{proof}[1][Proof]{\begin{trivlist}
  \item[\hskip \labelsep {\it #1:}]}{%
    \qed\end{trivlist}}

\usepackage{enumitem}
\setenumerate[1]{label=\rm(\it{\roman{*}}\rm),ref=({\it\roman{*}}),leftmargin=*}
\newlength{\wordlength}

\algsetup{linenodelimiter=\,}
\algsetup{linenosize=\tiny}
\algsetup{indent=2em}

			\newcommand{\pref}{\ensuremath{\succsim}}

\newcommand{\citet}[1]{\citeauthor{#1}~\shortcite{#1}}
\newcommand{\citep}{\cite}

\newcommand{\midd}{\mathbin{:}}

\begin{document}


\title{The Constrained Round Robin Algorithm for Fair and Efficient Allocation}


\author{
Haris Aziz\\
UNSW Sydney and Data61 CSIRO \\
haris.aziz@unsw.edu.au\\
\And
Xin Huang\\
The Chinese University of Hong Kong\\
xhuang@cse.cuhk.edu.hk 
\AND
Nicholas Mattei\\
Tulane University\\
nsmattei@tulane.edu 
\And
Erel Segal-Halevi\\
Ariel University\\
erelsgl@gmail.com
}


\maketitle

    \begin{abstract}
	    We consider a multi-agent resource allocation setting that models the assignment of papers to reviewers. A recurring issue in allocation problems is the compatibility of welfare/efficiency and fairness. Given an oracle to find a welfare-achieving allocation, we embed such an oracle into a flexible algorithm called the Constrained Round Robin (CRR) algorithm, that achieves the required welfare level. Our algorithm also allows the system designer to lower the welfare requirements in order to achieve a higher degree of fairness. If the welfare requirement is lowered enough, a strengthening of envy-freeness up to one item is guaranteed. Hence, our algorithm can be viewed as a  computationally efficient way to interpolate between welfare and approximate envy-freeness in  allocation problems.
\end{abstract}

\section{Introduction}
In multi-agent resource allocation, 
two orthogonal concerns are \emph{efficiency}, e.g., Pareto optimality or maximizing the sum of utilities, and \emph{fairness}, e.g., envy-freeness. In certain deployed applications, such as assigning papers to reviewers, the primary focus is on efficiency. But when the agents' preferences are highly correlated, there may be many different efficient allocations, most of which are not fair.  Indeed, fairness concerns pose their own challenges. In several settings concerning indivisible resources, a fair allocation may not exist, and if it does exist --- it might be hard to compute and might fail to satisfy efficiency requirements.

Several recent works focus on a particular combination of efficiency and fairness --- finding an allocation that is both Pareto-optimal (PO) and \emph{envy-free except one item (EF1)}.
\citet{CKM+16a} prove that any allocation maximizing the \emph{Nash welfare} (the product of utilities) is both PO and EF1, and \citet{BMV17a} present a pseudo-polynomial-time algorithm for finding a PO and EF1 allocation. 
These fascinating results are still limited in two ways: 
(1) They do not provide any guarantee for more popular welfare goals such as \emph{utilitarian maximality} (maximizing the sum of utilities) or \emph{rank-maximality}.
(2) They work only for the basic setting in which all utilities are weakly-positive (i.e., all items are goods), there is a single copy of each item, and there are no constraints on the number of items given to each agent.
Thus they do not apply to practical problems such as assigning conference-papers to reviewers in which there are capacity constraints on both reviewers and papers and where agents may have negative utilities for the tasks.

The central research question we examine is the following one: \emph{``how can we find fair allocations that guarantee a specified level of welfare/efficiency?''}




\paragraph{Our Contribution.}

We consider a general allocation problem in which the utility of each item for each agent can be any real number (positive, negative, or zero), and there may be capacity constraints on both agents and items.
We focus on EF1 as the predominant fairness notion.
We consider \emph{utilitarian-maximality} (UM) and \emph{rank-maximality} (RM) as the predominant efficiency notions, though our approach can handle other efficiency notions. 

Ideally, given a fairness and an efficiency requirement, we would like to first decide whether there exists an allocation that satisfies both requirements. If no such allocation exists, we would like to find a fair allocation that is as efficient as possible. However, we show that both these goals are computationally hard even for three agents:
(1) deciding whether there exists an allocation that is both EF1 and UM, or both EF1 and 
RM, is NP-complete;
(2) computing an allocation that is UM or RM within the set of EF1 allocations is NP-hard.

In view of these issues, we present a general algorithm that finds an efficient allocation in a way that is provably fair for important classes of utility functions such as identical utilities. 
It is inspired by the \emph{round robin} (RR)  protocol for item allocation: In the RR protocol, each agent gets turns in a round-robin manner and in turn picks an item.
Our algorithm is called \emph{Constrained Round Robin (CRR)}:
it tries to find an allocation in a round-robin manner, under the constraint that the final allocation must satisfy a given efficiency goal.
For a given efficiency goal W, W-CRR is the corresponding CRR algorithm specified with respect to W and it returns an allocation satisfying W. 
In the special case in which W is null (i.e, no welfare goal is imposed), we show that W-CRR finds an EF1 allocation for all utilities consistent with the ordinal preferences.
We call this novel strengthening of EF1 \emph{necessary EF1 (NEF1)}. 
\ifdefined\JOURNALVERSION
The concept is novel and is of independent interest.
It is an EF1 property that is robust to perturbations in the utility function.
\fi

The generality of CRR allows the system designer to gradually increase the welfare requirements by modifying W, as long as the resulting allocation is sufficiently fair for his/her needs%
\ifdefined\JOURNALVERSION
(see Figure~\ref{figure:RB})%
\fi
.
W-CRR can be considered a practical engineering approach for interpolating between welfare and fairness, re-using existing algorithms and implementations for finding allocations that achieve a target welfare level.
Using extensive simulation experiments on both real-world and synthetic data, we examine W-CRR where W is utilitarian-maximality or rank-maximality, and show that in many settings W-CRR produces fairer allocations than several existing algorithms for maximum Nash welfare or maximum Egalitarian welfare.


%


\ifdefined\JOURNALVERSION
\begin{figure}
\smartdiagram[circular diagram:clockwise]{reduce W if the allocation is not fair enough, W-CRR}
\caption{Our framework of finding an allocation satisfying a target welfare level W and target level of envy-freeness up to one item fairness.}
\label{figure:RB}
\end{figure}
\fi

\section{Related Work}
\citet{BCM15a} survey the main algorithms and considerations in fair item allocation, including fairness and welfare considerations.

A flexible approach to achieve varying levels of utilitarian and egalitarian welfare is applying ordered weighted averages over the agents' total utilities~\citep{Yager88a}. The approach smoothly interpolates between utilitarian welfare and egalitarian welfare but it is computationally intractable. There is also no known connection with EF1. 

Nash welfare based approaches have been proposed to achieve both Pareto optimality and EF1~\citep{CKM+16a}, but computing it is NP-hard. 
When there are additional side constraints such as agents getting an equal number of items, then EF1 is not implied by an allocation maximizing Nash welfare. 
\citet{LMNW18a} proposed the Sum-OWA Maximizing rule that uses ordered weighted averaging over the \emph{agent's individual utilities for items}. 
\ifdefined\JOURNALVERSION
The idea behind the approach is to put less emphasis on an agent getting lesser preferred items if she has already got more preferred items thereby letting other agents have more chance of getting their most preferred items. 
\fi
They experimentally showed that the approach can reduce certain measures of inequality. They did not establish whether the approach achieves some formal notion of fairness for any of the ordered vectors. The papers discussed assume that the agents' utilities are positive. On the other hand, W-CRR, in general, can cater for positive or negative utilities. 
\ifdefined\JOURNALVERSION
The approach of  \citet{LMNW18a} also guarantees a weak form of Pareto optimality. On the other hand, we can guarantee Pareto optimality 
\fi

\citet{GKKM+10a} considered the problem of assigning papers to referees and focused on maximizing the utility of the worst off agent. They also considered rank maximal allocations when considering reviewer assignments. Rank maximality has been considered as a criterion when computing desirable allocations~\citep{Palu13a}.
\ifdefined\JOURNALVERSION
In recent years, polynomial-time combinatorial algorithms have been presented for many-to-many allocations with upper capacity constraints~\citep{Palu13a}.
\fi

\citet{BiBa18a} considered fairness issues when there are upper cardinality constraints on the categories of items in agents' allocations. They did not consider welfare measures. 

\citet{SAH17a} considered a sub-class of additive valuations and characterized the conditions under which allocations are necessarily envy-free for all valuations in the sub-class. 
\citet{Hale18b} considered fairness-welfare interpolation in the division of a heterogeneous resource (``cake-cutting''). The results crucially depend on the divisibility of the cake. 
\ifdefined\JOURNALVERSION
He presented an algorithm that, given an efficient allocation, a fair allocation, and a trade-off ratio $r\in[0,1]$, returns an allocation that is an $r$-approximation to  efficiency and $(1-r)$-approximation to fairness. However, the algorithm crucially depend on the divisibility of the cake. 
It is not clear whether it can be generalized to indivisible goods.
\fi
\citet{CKKK12} studied the fairness-welfare tradeoff through the lens of the \emph{price of fairness} --- the ratio between the maximum welfare of an arbitrary allocation and the maximum welfare of a fair allocation.
\ifdefined\JOURNALVERSION
For indivisible items, they prove that the price of proportionality is $n-1+1/n$, the price of envy-freeness is between $n/3+7/9$ and $n-1/2$, and the price of equitability might be infinite.
Their upper bounds are not constructive --- they prove that \emph{every} proportional/envy-free allocation has a price-of-fairness of at most $n-O(1)$.
\fi

The complexity of computing a Pareto optimal and EF1 allocation is an open question in the literature~\citep{CKM+16a,BMV17a}. We show that strengthening PO to maximum welfare immediately leads to computational hardness.
	
\section{Setting}
\label{sec:setting}
\subsection{Agents and Utilities}
An allocation problem is a tuple $(N,O,\pref,u, c)$ such that $N=\{1,\ldots, n\}$ is a set of agents, $O=\{o_1,\ldots, o_m\}$ is a set of items and the preference profile $\pref=(\pref_1,\ldots, \pref_n)$ specifies for each agent $i$ his complete and transitive preference $\pref_i$ over $O$. Agents may be indifferent among items, i.e., they provide weak rankings of the items.  Hence, we denote $\pref_i: E_i^1,\dots,E_i^{k_i}$ for each agent $i$ with equivalence classes in decreasing order of preferences.
%
%
We assume that each $\pref_i$ is induced by a utility function $u_i: O\rightarrow \mathbb{R}$. 
In general the utility of an item can be positive, negative or zero. Some of our results apply only when all utilities have the same sign (all positive or all negative).
Agents are assumed to have additive utility so that $u_i(O')=\sum_{o\in O'}u_i(o)$. 

When $u_i$ is not known to the algorithm (for example, when the user interface does not allow the agents to submit their cardinal utilities), we can still take a cardinal approach by assuming that agent's utility functions are derived by some \emph{scoring rule}, that ascribes a utility for an item being in a given rank in the preference list. For example, in \emph{Borda scoring}, an agent with strict preferences over $m$ items gives utility $m-i+1$ for the $i$-th ranked item (see e.g., \citep{BLN+17a}).

An {\em allocation} $p$ is a function $p:N\rightarrow 2^O$ assigning each agent a bundle of items. Typically it is assumed that each item is allocated exactly once. However, we do not impose such a restriction, since we want to capture scenarios such as reviewer assignment, in which a paper is reviewed by multiple agents.
\ifdefined\JOURNALVERSION
agents have lower and upper bounds for how many papers they get, and papers have lower and upper bounds for how many agents review a paper.
Note that this cannot be captured by just creating several identical copies of the same item, since this would allow to give two copies of the same item to the same agent, which is usually not allowed.
\fi

There is a \emph{capacity function} $c$, which specifies for each agent $i$, the values $\underline{c}_i$ and $\overline{c}_i$, which are the minimum and maximum number of items that should be allocated to $i$.
Similarly, $c$ specifies for each item $o$, the
values $\underline{c}_o$ and $\overline{c}_o$, which are the minimum and maximum number of agents it may be allocated to.
Some of our results focus on the special case where for all $o,o'\in O$, $\overline{c}_o=\underline{c}_{o}=\overline{c}_{o'}=\underline{c}_{o'}$, 
and for all $i,i'\in N$, $\overline{c}_i=\overline{c}_{i'}$ and $\underline{c}_i=\underline{c}_{i'}$ and $\overline{c}_i-\underline{c}_i=1$. The assumptions are motivated by the conference reviewer assignment problem in which each paper is reviewed by the same number of agents and all reviewers review almost the same number of papers.  


We denote the set of all feasible allocations (all allocations consistent with the capacity constraints) by $\mathcal{A}$.
An allocation $p\in \mathcal{A}$ is called \emph{balanced} if $|(|p(i)|-|p(j)|)|\leq 1$ for all $i,j\in N$.

%
%

The set of all utility functions consistent with $\pref_i$ is denoted by $\mathcal{U}(\pref_i)$. 	
In this paper we do not consider strategic manipulations --- we assume that all agents truthfully report their valuations.

\ifdefined\JOURNALVERSION
\begin{example}
\label{exm:identical}
 Consider the scenario where $N=\{1,2,3,4,5,6\}$, $O=\{o_1,o_2, o_3,o_4\}$, and the preferences of all agents are the same: $o_1 \succ o_2 \succ o_3 \succ o_4$. Each paper has to be reviewed thrice and each agent can review at most two papers. So for each object $o\in O$, $\underline{c}_o = \overline{c}_o = 3$, and for each agent $i\in N$, $\underline{c}_i = \overline{c}_i = 2$.
 An example of a feasible allocation is one in which agents $1,2,3$ each get $\{o_1,o_2\}$ and $4,5,6$ each get $\{o_3,o_4\}$.
\end{example}
\fi

\subsection{Fairness}
An allocation $p$ is \emph{envy-free} (EF) if 
$u_i(p(i))\geq u_i(p(j))$ for all $i,j\in N$. Although EF is often viewed as the gold-standard for fairness, there may not exist an allocation that is EF. Furthermore, checking whether such an allocation exists is NP-complete (see e.g., \citep{BoLe15a}). 

An allocation $p$ is \emph{envy-free up to one item (EF1)} if for all $i,j \in N$, there exists a subset $O'\subseteq O$ of cardinality at most $1$, such that either $u_i(p(i)\setminus O') \geq u_i(p(j))$ or $u_i(p(i)) \geq u_i(p(j)\setminus O')$. Note that this definition makes sense for both positive and negative utilities.

An allocation $p$ satisfies \emph{necessary envy-freeness up to one item (NEF1)}  if
for all $i,j \in N$, 
there exists a subset $O'\subseteq O$ of cardinality at most $1$, such that 
for all $u_i\in \mathcal{U}(\pref_i)$, 
either $u_i(p(i)\setminus O') \geq u_i(p(j))$ or $u_i(p(i)) \geq u_i(p(j)\setminus O')$.

There is an equivalent way to write the definition of NEF1.
In the \emph{responsive set extension}, preferences over items are extended to preferences over sets of items in such a way that a set in which an item is replaced by a more preferred item is more preferred. For two allocations $p(i)$ and $q(i)$ of equal size,  we write $p(i) \mathrel{\pref_i^{RS}} q(i)$ if there is a bijection $f$ from $q(i)$ to $p(i)$ such that for each $o\in q(i)$, $f(o)\pref_i o$.
We can then say that an allocation $p$ is NEF1 if 
for all $i,j \in N$, there exists a subset $O'\subseteq O$ of cardinality at most $1$, such that
either 
$p(i)\setminus O' \pref_i^{RS} p(j)$
or
$p(i)\pref_i^{RS} p(j)\setminus O' $.
The equivalence of definitions is based on results by ~\citet{AGMW15a}.
NEF1 is a relaxation of the NEF concept as considered in a series of papers on fair allocation of indivisible items~\citep{AGMW15a,BEL10a,BEF03a}.
	
\subsection{Welfare}
There are many notions of welfare. We particularly focus on the following two.
An allocation $p$ is \emph{utilitarian-maximal (UM)} if it maximizes the sum of utilities:
$
p\in \arg\max_{q\in \mathcal{A}}\sum_{i\in N}\sum_{o\in q(i)}u_i(o)
$.
An allocation $p$ is \emph{rank-maximal (RM)} if it maximizes the number of agents assigned their 1st-best item, subject to that --- it maximizes the number of agents assigned their 2nd-best item, and so on (see e.g., \citep{irving2006rank,Manl13a}). 
Formally,  for any item $o\in O$, its \emph{rank} in agent $i$'s preference list $\pref_i$ is $j$ if $o\in E_i^j$ i.e., it is in $i$'s $j$-th equivalence class. For an allocation $p$, its corresponding \emph{rank vector} is $r(p)=(r_1(p), \ldots, r_m(p))$ where $r_j(p)$ is the total number of 
times an item $o$ is assigned to an agent for whom $o$ has a rank of $j$. Formally:
$
r_j(p) := \sum_{i\in N}|E_i^j\cap p(i)|.
$
We compare rank vectors lexicographically. A rank-vector $r=(r_1,\ldots, r_m)$ is \emph{better} than another rank-vector $r'=(r_1',\ldots, r_m')$ if, for the smallest $i$ such that $r_i\neq r_i'$, it must hold that $r_i>r_i'$.
An allocation $p$ is \emph{rank-maximal} if:
$p\in \arg\max_{q\in \mathcal{A}}r(q).$

We note that rank maximality can be captured by imposing specific cardinal utilities for agents and then maximizing utility~\citep{Mich07a}. The specific cardinal utilities are lexicographic: For an agent $i$ getting an item $o$ from the $j$-th equivalence class, she gets utility $({xm})^{m-j}$ where $x=\max\{\overline{c}_o\midd o\in O\}$.
 
\section{Combining Welfare and  Fairness}
\label{sec:combining}
Before presenting our W-CRR algorithm, we explain the rationale behind it. Our starting point is the observation that there is a simple class of algorithms that always return NEF1 allocations --- the \emph{recursively balanced sequential allocation}. These sequential allocation algorithms are parametrized by a pre-specified sequence of agent-turns. In sequential allocation, when an agent gets a turn, she picks a most preferred item that is not yet allocated~(see e.g., \citep{AWX15b,BoLa11a,BBL+18a}). In case of our general setting with capacity constraints, an agent picks a most preferred item that she can pick without violating the item capacity constraint. 
A sequence of turns is called \emph{recursively balanced (RB)} if at no point an agent has taken more than one turn than another agent (see e.g., \citep{AWX15b}).
For example for two agents and 4 items, the RB sequences are 1212, 1221, 2121 and 2112. 
We are also interested in the following subset of RB sequences. A sequence of turns is called \emph{round-robin (RR)} if, for some permutation $\pi$ of the agents, the sequence is $\pi_1,\pi_2,\ldots,\pi_n,\pi_1,\pi_2,\ldots,\pi_n,\ldots$
For example, for two agents and 4 items the RR sequences are 1212 and 2121.
All sequential allocation algorithms are \emph{ordinal} in nature because they ignore the cardinal information in the utility function of the agents. 
\ifdefined\JOURNALVERSION
They simply used the relative rankings of the items in the agents' preference lists. 

\begin{example}
Consider the following preferences of the agents.
\begin{itemize}
\item Agent 1: $a \succ b \succ c \succ d$
\item Agent 2: $b \succ c \succ d \succ a$
\end{itemize}
A recursively balanced picking sequences include 1212, 1221, 2121 and 2112. Suppose we use the sequence 1212. Then in the final allocation agent 1 gets $\{a,c\}$ and agent 2 gets $\{b,d\}$. Agent 1 will first pick $a$, then 2 $b$, then 1 $c$, and 2 $d$ (recall that we assume the agents are truthful).
\end{example}
\fi

Note that an RB picking-sequence always returns a balanced allocation. Therefore, if the capacity constraints rule out some balanced allocations (e.g. one agent must get exactly 3 items while another agent must get exactly 5 items), then RB may not return a feasible allocation. However, it always returns an approximately-fair allocation:
\begin{proposition}
\label{rb-is-nef1}
(a)
If all utilities are positive, then the outcome of sequential allocation for any \emph{RB} sequence is NEF1. 
(b)
If all utilities are negative, then the outcome of sequential allocation for any \emph{RR} sequence is NEF1.
\end{proposition}
\begin{proof}
(a) 
Let agent $i$'s allocation be $s_1,\ldots, s_k$ and agent $j$'s allocation be $t_1,\ldots, t_{k'}$. Since the picking sequence is recursively balanced, the allocation is balanced, which implies that $k'\in \{k-1,k,k+1\}$. Without loss of generality we let $k'=k+1$ by letting $j$ get dummy items after her actual turns have finished. Now let us look at the sequence of turns in which the items are picked: 
$ \{s_1,t_1\}, \ldots, \{s_k, t_k\}, \{t_{k+1}\}.$
When two items are paired in curly brackets, we do not distinguish which one is picked first. Note that for each $j\in \{2,\ldots, k'\}$, $s_{j-1}\pref_i t_{j}$ because $i$ picks $s_{j-1}$ when $t_j$ is still available. Consider the function $f$ from $p(j)\setminus \{t_1\}$ to $p(i)$ such that for each $t_j\in p(j)\setminus \{t_1\}$, $f(o)=s_{j-1}$. Hence $p(i)\pref_i^{RS} p(j)\setminus \{t_1\}$ and therefore, $p$ satisfies NEF1. 

%

(b) The argument is similar. It works by removing the most undesirable item allocated to the envious agent, instead of the most desirable item allocated to the envied agent. 
\end{proof}
Note that if all utilities are negative, the outcome of an RB sequence that is \emph{not} RR might not even be EF1 (consider the sequence A,B,B). 
If utilities are mixed, then even the outcome of an RR sequence might not be EF1 \citep{ACI+18}.


Sequential allocation with respect to a round-robin sequence does not necessarily maximize any social welfare function.
Moreover, even EF1, which is weaker than NEF1, might be incompatible with welfare-maximization requirements. It is well-known that EF1 is incompatible with UM when there are no capacity constraints (for example, if Alice's utility for every item is higher than Bob's utility, the only UM allocation gives all items to Alice, which is obviously not EF1 for Bob). The following example shows that this incompatibility also holds when all agents have an identical capacity constraint and when all utilities are based on Borda scores, and it also holds for RM. 
\ifdefined\JOURNALVERSION
\else
\fi
\begin{example}
There are 3 agents and 9 items, with a single copy per item. The following table shows the utility of each item to each agent:

\noindent
\begin{tabular}{cccccccccc}
\hline 
	\footnotesize
\textbf{Item}: & $o_1$ & $o_2$ & $o_3$ & $o_4$ & $o_5$ & $o_6$ & $o_7$ & $o_8$ & $o_9$
 \\
 \hline
\footnotesize
\textbf{Agents 1,2}: & 9 & 8 & 7 & 6 & 5 & 4 & 3 & 2 &
1
\\ 
\hline
\footnotesize
\textbf{Agent 3}: & 6 & 9 & 8 & 7 & 5 & 4 & 3 & 2 & 1
\\ 
\hline 
\end{tabular} 
We first consider the case in which each agent's capacity is unbounded. In any utilitarian-maximal or rank-maximal  allocation, each item must be given to an agent who assigns to it a highest utility. Therefore, Agent 3 must get (at least) the set $\{o_2,o_3,o_4\}$.
This leaves the other two agents a total utility of at most $24$. So one of these agents gets a utility of at most $12$. This agent EF1-envies Agent 3.

Next, we consider the case in which each agent must get exactly 3 items (i.e., the allocation must be balanced). Under this constraint,  utilitarian-maximal or rank-maximal imply Agent 3 must get exactly $\{o_2,o_3,o_4\}$.
By the same argument, no utilitarian-maximal or rank-maximal  allocation is EF1.
\end{example}

 
Given this incompatibility, we would like to decide whether there exists, among all welfare-maximizing allocations, one that is also fair. However, this can be computationally challenging. 
Below we show that, checking whether a UM / RM and EF1 allocation exists is NP-complete when there are at least three agents.

%
%
%
%

%
\begin{proposition}
For three agents, checking whether there exists an allocation that is both rank-maximal and EF1, or both utilitarian-maximal and EF1, is NP-complete.
\end{proposition}
\begin{proof}
We reduce from \textsc{Partition}, which is the following decision problem: given $m$ numbers $a_1,\ldots,a_m$ summing up to $2S$, is there a partition of the numbers into two sets with each set of integers summing up to $S$?
Given an instance of \textsc{Partition}, define an item-allocation instance with $m+6$ items: $m$ number-items and 6 extra-items, denoted $e_1,e_2,e_3,e_4,e_5,e_6$. 
Each item should be allocated once. There are $3$ agents with unbounded capacities and the following valuations:
\noindent
\begin{center}
\begin{tabular}{ccc}
\hline 
 Items:   & $m$ number-items & $e_1,e_2,e_3,e_4,e_5,e_6$ \\ 
\hline 
Alice:     & 0 & $5S,4S,3S,2S,S,0$ \\ 
\hline 
Bob, Carl: & $u_j(o_j) = a_j$  & $4S,3S,2S,S,5S,5S$ \\ 
\hline 
\end{tabular} 
\end{center} 

It is possible to give all agents their best item, so any RM allocation must do so and give $e_1,e_5,e_6$ to Alice, Bob and Carl respectively.
Subject to this, it is possible to give one agent her 2nd-best item --- this must be Alice, so she must be given $e_2$. Similarly, Alice must be given $e_3$ and $e_4$. The same is true for any UM allocation.

Now, Bob and Carl have utility $5S$, and they value Alice's bundle as $10S$.
Since the highest item-value for both Bob and Carl is $4S$, the allocation becomes EF1 iff each of them gets a utility of exactly $6S$, iff each of them gets a utility of exactly $S$ from the number-items, iff the allocation of these items corresponds to an equal partition of the numbers.
\end{proof}

\begin{remark}
When there are two agents, checking whether there exists an allocation that is both UM and EF1 is polynomial-time solvable. 
\end{remark}

\begin{proposition}\label{prop:ef1w}
For three agents, computing an allocation that is utilitarian-maximal or rank-maximal within the set of EF1 allocations is NP-hard. 
\end{proposition}
\begin{proof}
Suppose there exists a polynomial-time algorithm that computes an allocation maximizing the utilitarian welfare within the set of EF1 allocations. Then we can compute the outcome of the algorithm and check whether its total utilitarian welfare is equal to the maximum possible welfare (which is well-known to be polynomial-time computable).
The two values are equal if and only if there exists an allocation that is both UM and EF1. An analogous argument holds for RM and EF1, where instead of the maximum utility we consider the maximum rank-vector. 
\end{proof}

The NP-hardness of computing an allocation that maximizes the utilitarian welfare within the set of EF1 allocations was proven independently and recently by \citet{BJK+19a}. Proposition~\ref{prop:ef1w} is stronger in one respect that it even holds for rank maximality.

We proved that deciding whether there exists a UM+EF1 allocation for 3 or more agent is NP-hard. The same is true if we replace EF1 with the more ``friendly'' concept of PROP1 --- proportionality up to at most a single item. 
An allocation $\pi$ satisfies \emph{proportionality up to one item (PROP1)} if for each agent $i\in N$, 
\begin{itemize} 
\item $u_i(\pi(i))\geq u_{i}(O)/n$; or
\item $u_i(\pi(i))+u_i(o)\geq u_{i}(O)/n$ for some $o\in O\setminus \pi(i)$; or 
\item $u_i(\pi(i))-u_i(o)\geq u_{i}(O)/n$ for some $o\in \pi(i)$. 
\end{itemize}

 \begin{proposition}
 When there are at least three agents, checking whether there exists an allocation that is both UM and PROP1 is NP-complete.
 \end{proposition}
  \begin{proof}
 Similarly to the proof for UM+EF1, we reduce from \textsc{Partition}. Given an instance of \textsc{Partition} with $m$ numbers, define an item-allocation instance with
 $m+4$ items: $m$ number-items plus $4$ extra-items.
 The agents' valuations are:
 \noindent
 \begin{center}
 \begin{tabular}{|c|c|c|}
 \hline 
 Items:          & $m$ number-items & $e_1, e_2, e_3, e_4$ \\ 
 \hline 
 Alice:     & 0 & $3S/2, 3S/2, 3S/2, 3S/2$ \\ 
 \hline 
 Bob, Carl: & $v(o_i) = a_i$ & $S, S, S, S$ \\ 
 \hline 
 \end{tabular} 
 \end{center}
 All agents value the set of all items as $6 S$, so the proportional share per agent is $2 S$. Since the largest item-value for both Bob and Carl is $S$, the allocation is PROP1 iff both Bob and Carl get at least $S$ from the number-items, iff the iff the allocation of these items corresponds to an equal partition of the numbers.
  \end{proof}

\section{The W-CRR Algorithm}
\label{sec:CRR}
In this section, we present the W-Constrained Round Robin (W-CRR) Algorithm. It is parametrized by a target welfare measure W.

While we are particularly interested in the case in which W is UM or RM, we note that W can be any other welfare notion, for example: maximum egalitarian welfare (Egal), maximum Nash welfare (Nash), utilitarian welfare that is at least some given threshold, Pareto optimality, etc.
The idea of W-CRR is to combine the required welfare guarantee with a round-robin sequential allocation.

%
%

	
The algorithm \emph{tries} to find an allocation as a result of some recursively-balanced and if possible a round-robin policy. The policy is left unspecified at the start of the algorithm and is searched for dynamically. While trying to obtain an allocation as a result of an RB sequence, the algorithm makes regular checks on whether the partial allocation obtained so far can be completed so as to satisfy W. If this is not possible, then the algorithm carefully relaxes the requirement of a recursively-balanced picking sequence. In this relaxation, the algorithm prioritizes agents with smaller allocations to make a pick. 
	
The algorithm uses an oracle called \emph{W-Completion} that takes as input a partial allocation $p$ and finds a feasible complete allocation that extends $p$ and satisfies W, or reports that no such extension exists. The oracle solves the following computational problem, which we also refer to as W-Completion.

\pbDef{{\footnotesize W-Completion}}
{\footnotesize $(N,O,\pref,u,c)$ and a partial allocation $p$}
{\footnotesize Does there exist a complete and feasible allocation $p'$ that extends $p$ and satisfies W? (yes/no)}

For several efficiency measures such as rank-maximality and utilitarian-maximality, {W-Completion} can be reduced to simply finding an allocation which satisfies W; see Lemma \ref{lem:reduction}.
If W is null, that is we do not impose any efficiency constraints, then W-CRR reduces to round-robin sequential allocation. 
The W-CRR algorithm is specified as Algorithm~\ref{algo:CRR}.
\ifdefined\JOURNALVERSION
Here is an illustrative example for how the algorithm works.
\begin{example}[W-CRR]
Consider an instance with 4 agents and 6 items, where each item should be assigned to two different agents and each agent gets three items. The following table shows the utility of each item to each agent:
\begin{tabular}{ccccccc}
\hline
\textbf{Item}: & $o_1$ & $o_2$ & $o_3$ & $o_4$ & $o_5$ & $o_6$ \\
\hline
\textbf{Agents 1,2,3}: & 6 & 5 & 4 & 3 & 2 & 1 \\
\hline
\textbf{Agent 4}: & 2 & 6 & 5 & 4 & 3 & 1 \\
\hline
\end{tabular}

Let the algorithm target at maximum utility i.e. $W$ is maximum total utility. For simplicity, we assume that the round-robin order is from agent $1$ to $4$ if there is a tie. Then the flow of algorithm is as following:
\begin{itemize}
\item Agent 1 gets $o_1$, Agent 2 gets $o_1$, Agent  3 gets $o_2$, Agent 4 gets $o_2$.
\item Agent 1 gets $o_3$, Agent 2 gets $o_3$, Agent  3  gets $o_4$, Agent  4 gets $o_4$.
\item Agent 1 gets $o_5$, Agent 4 gets $o_5$ (at this time, Agents 2 and 3 cannot get their first choice $o_5$, otherwise the maximum utility constraint would be broken), Agent 2 gets $o_6$, Agent 3 gets $o_6$.
\end{itemize}
The final allocation is EF1, while maximizing total utility subject to the capacity constraints.
\else
\fi

Note that during the running of the algorithm W-CRR, as we grow allocation $p$, we always ensure that there exists an extension of $p$ that also satisfies W. W-CRR needs to call the W-Completion oracle once for each agent-item pair, for a total of $m n$ times. Hence the following properties are obvious:
\begin{proposition}
For any welfare measure W and instance $I=(N,O,\pref,u,c)$:
(1) If $I$ has at least one allocation that satisfies W, then W-CRR returns an outcome that satisfies W.
(2) If W-Completion can be solved in time $T(I)$, then W-CRR takes time $mn\cdot T(I)$. 
\end{proposition}


W-CRR is inspired by the sequential allocation mechanism that uses a round-robin picking sequence. Among the agents who have the least number of items, it tries to give an item to the agent with the lowest index.

Given an instance $I$, we denote by $W(I)$ the set of allocations that satisfy target $W$.
\begin{proposition}
\label{crr-is-rr}
If $W(I)$ contains the set of all balanced allocations,
then W-CRR is equivalent to a round-robin sequential allocation rule.
\end{proposition}
Combining Propositions \ref{crr-is-rr} and \ref{rb-is-nef1} implies that, if $W(I)$ contains all balanced allocations, and if all utilities are positive or all are negative, then W-CRR returns an allocation that is not just EF1 but also NEF1. 

	
%
%
	
\begin{algorithm}[h!]
 	\caption{W-CRR}
 	\label{algo:CRR}
\begin{algorithmic}
\REQUIRE $I=(N,O,\pref,u,c)$ and a target welfare measure $W$ that can be achieved by at least one feasible allocation.  
\ENSURE Allocation $p$
\end{algorithmic}
	\begin{algorithmic}[1]
	\STATE Initialize $p$ to empty allocation \COMMENT{$p$ will be gradually extended.}
	\STATE $j\longleftarrow 1$
\STATE $O'\longleftarrow O$.
\STATE For each $i\in N$, initialize $L_i'=({E_i^1}', {E_i^2}',...)$ to $L_i=(E_i^1,E_i^2,...)$
\STATE $N^{active}\longleftarrow N$ \COMMENT{$N^{active}$ is the set of agents with non-empty preferences lists $L_i'$s} 
\WHILE{$N^{active}\neq \emptyset$ and there exists a feasible extension of $p$ by giving a new item to some agent}
\STATE $N^{Min}\longleftarrow \{i\in N^{active}\midd \not\exists j\in N^{active} \text{ s.t. } |p(i)|>|p(j)|\}$ \\
\COMMENT{$N^{Min}$ is the set of agents among $N^{active}$ who have least number of items} 
\IF{there exists an $i\in N^{Min}$ and $o\in {E_i^1}'$ such that { W-Completion}$(p\cup \{(i,o)\})$ is yes \COMMENT{W-Completion takes into account capacity constraints $c$}} 
\STATE Choose such an $i$ with the lowest index. 
\STATE $p\longleftarrow p\cup \{(i,o)\}$
\COMMENT{Extend the current allocation by assigning item $o$ to agent $i$}
\ELSE 
\STATE For each $i\in N^{Min}$, remove ${E_i^1}'$ from $L_i'$ 
and rename the equivalence classes of $L_i'$ \COMMENT{rename so that the next equivalence class is now named ${E_i^1}'$ and so on.}
\STATE $N^{active}\longleftarrow \{i\in N^{active}\midd L_i'\neq \text{empty}\}$ 
\ENDIF
	 	\ENDWHILE
	\RETURN $p$.
\end{algorithmic}
\end{algorithm}

 \begin{example}[Example illustrating W-CRR]
 Consider an instance with 4 agents and 6 items, where each item should be assigned to two different agents and each agent gets three items. The following table shows the utility of each item to each agent:

 \begin{tabular}{ccccccc}
 \hline 
 \textbf{Item}: & $o_1$ & $o_2$ & $o_3$ & $o_4$ & $o_5$ & $o_6$ \\
 \hline 
 \textbf{Agents 1,2,3}: & 6 & 5 & 4 & 3 & 2 & 1 \\ 
 \hline 
 \textbf{Agent 4}: & 2 & 6 & 5 & 4 & 3 & 1 \\ 
 \hline 
 \end{tabular}

 Let the algorithm target at maximum utility i.e. $W$ is maximum total utility. For simplicity, we assume that the round-robin order is from agent $1$ to $4$ if there is a tie. Then the flow of algorithm is as following:
 \begin{itemize}
 \item Agent 1 gets $o_1$, Agent 2 gets $o_1$, Agent  3 gets $o_2$, Agent 4 gets $o_2$.
 \item Agent 1 gets $o_3$, Agent 2 gets $o_3$, Agent  3  gets $o_4$, Agent  4 gets $o_4$.
 \item Agent 1 gets $o_5$, Agent 4 gets $o_5$ (at this time, Agents 2 and 3 cannot get their first choice $o_5$, otherwise the maximum utility constraint would be broken), Agent 2 gets $o_6$, Agent 3 gets $o_6$.
 \end{itemize}
 The final allocation is EF1, while maximizing total utility.
 \end{example}

\subsection*{RM-CRR and UM-CRR}
Since the W-CRR algorithm depends on a W-Completion oracle, it is interesting to know for what W such W-Completion oracles can be implemented efficiently. 
We now show two cases in which W-Completion can be reduced to just computing a single allocation that satisfies W: they are rank-maximality and utilitarian-maximality. 
Rank-maximal allocations with arbitrary upper capacities can be computed efficiently by polynomial-time combinatorial algorithms \citep{Palu13a}.
Similarly, it is well-known that utilitarian maximal allocations can be computed in polynomial time even with capacity constraints~(see e.g., \citep{KoVy06a}). 
\begin{lemma}
\label{lem:reduction}
There exists a linear-time reduction -- 
(a) from {UM-Completion} to computing a utilitarian-maximal allocation;
(b) from {RM-Completion} to computing a rank-maximal allocation.
\end{lemma}
\begin{proof}
\textbf{For part (a)}, 
consider a partial allocation $p$ that respects the upper capacities. 
Based on $p$, decrease the upper and lower capacities, and find a utilitarian-maximal allocation $q$ for the decreased capacities.
Let $r$ be the allocation that combines $p$ and $q$. 
Let $w_p, w_q, w_r$ be the utilitarian welfare of $p,q,r$ respectively, so that $w_r = w_p+w_q$.
Let $w_*$ be a maximal utilitarian welfare in the entire instance (note that $w_*$ can be calculated once per instance --- it need not be calculated in each call to UM-Completion).

If $w_r \geq w_*$. then UM-Completion returns ``yes'', since $r$ satisfies the (original) capacity constraints and attains the maximum utilitarian welfare.
If $w_r < w_*$ then UM-Completion returns ``no'', 
since if there were a completion $q'$ of $p$ 
such that the completed allocation is utilitarian-maximal, then we would have $w_p + w_{q'}\geq w_*  > w_r = w_p + w_{q}$, which would imply $w_{q'} > w_q$, contradicting the maximality of $q$ with the decreased capacities.

\textbf{For part (b)}, 
note that a rank-maximal allocation is just like a utilitarian-maximal allocation with a vector-valued welfare-function. 
%
\end{proof}	

From the lemma above, we get the following propositions.
\begin{proposition}
{W-Completion} can be solved in polynomial time when W is (a) utilitarian-maximality, or (b) rank-maximality.
{W-CRR} can be computed in polynomial time when W is (a) utilitarian-maximality, or (b) rank-maximality.
\end{proposition}

We note that CRR aligns well with Borda welfare. 

\begin{proposition}
Under strict preferences, positive Borda utilities and items having identical capacities, if the set of feasible allocations contains all balanced allocations, then W-CRR for any target social welfare $W$ gives a 2-approximation of maximum total Borda welfare. 
\end{proposition}
\begin{proof}
We first prove that under strict preferences, the output of any RB sequential allocation gives a 2-approximation of maximum total Borda welfare. 
Let $c$ be a positive integer. Suppose there are $c$ copies of each item. We prove that, in round-robin, the welfare in the $k$-th round increases by at least $m-\lfloor (k-1)/c\rfloor$. Then the total utility is at least $c(m+1)m/2$. 
To prove this, let us consider the worst case for the $k$-th round. Suppose in the $k$-th round, agent $i$ is the one to choose his best item. Then in agent $i$'s perspective, the worst case is that in each of the previous rounds, one of his highest-utility item were taken. In such situation, agent $i$ gets utility $m-\lfloor(k-1)/c\rfloor$.
Obviously, the maximum welfare is at most $cm^2$, which is attained when in each round the utility increases by $m$. Combining these, we conclude that RB sequential allocation gives a 2-approximation of Borda welfare. 

Now consider W-CRR for any target social welfare welfare $W$
If W is less than half the maximum total Borda welfare, then CRR degenerates to any RB sequential allocation. If the target welfare is more than half the maximum total Borda welfare, then this target welfare is indeed achieved by CRR. 
\end{proof}

\section{Empirical Evaluation}
\label{sec:empirical}
In order to evaluate the performance of W-CRR, we have implemented it along with objectives UM, RM, Egal, Nash,  and $\sum$-OWA Maximizing.
In order to easily switch between the various objective functions all models were implemented in Gurobi 8.1 using the Python interface.  All code from this project will be released on GitHub.  We have run experiments on both real-world data from PrefLib \cite{MaWa17,MaWa13a} and data generated according to a Mallows model.  

In all of our experiments, we assume that agents have Borda utilities and these are assigned at the level of the equivalence class in the case of weak rankings from the realworld data; this is common in much of the literature on resource allocation and conference paper assignment \cite{LMNW18a}. In all of our experiments we enforce that all agents and all items, respectively, have the same maximum and minimum capacities. We also enforce that all agents may receive only one copy of any particular item.  
For each of the datasets, we report the fraction of the $n(n-1)$ relations between agents that satisfy EF, EF1, NEF, and NEF1.

\subsection{Synthetic Data}

We generate synthetic data using a Mallows model which gives each agent a strict preference over the items and allows us to measure the effect that correlation between the agent preferences has on the overall results.  For each setting of the $\phi\in [0,1]$ parameter, we generate 25 test instances and report the mean and standard deviation of the proportion of allocations satisfying the various evaluation criteria. We tested for correlation $\phi \in \{ 0.0, 0.25, 0.50, 0.75, 0.95\}$.  
\emph{The value of $1-\phi$ indicates the level of correlation.}
For each $\phi$, we generated instances with 10 agents and 20 items. We set the agent minimum to 3 and maximum to 6, while for the objects we set a minimum of 3 and maximum of 4.

\begin{figure}
 	\centering
 	\includegraphics[width=0.6\linewidth,page=4]{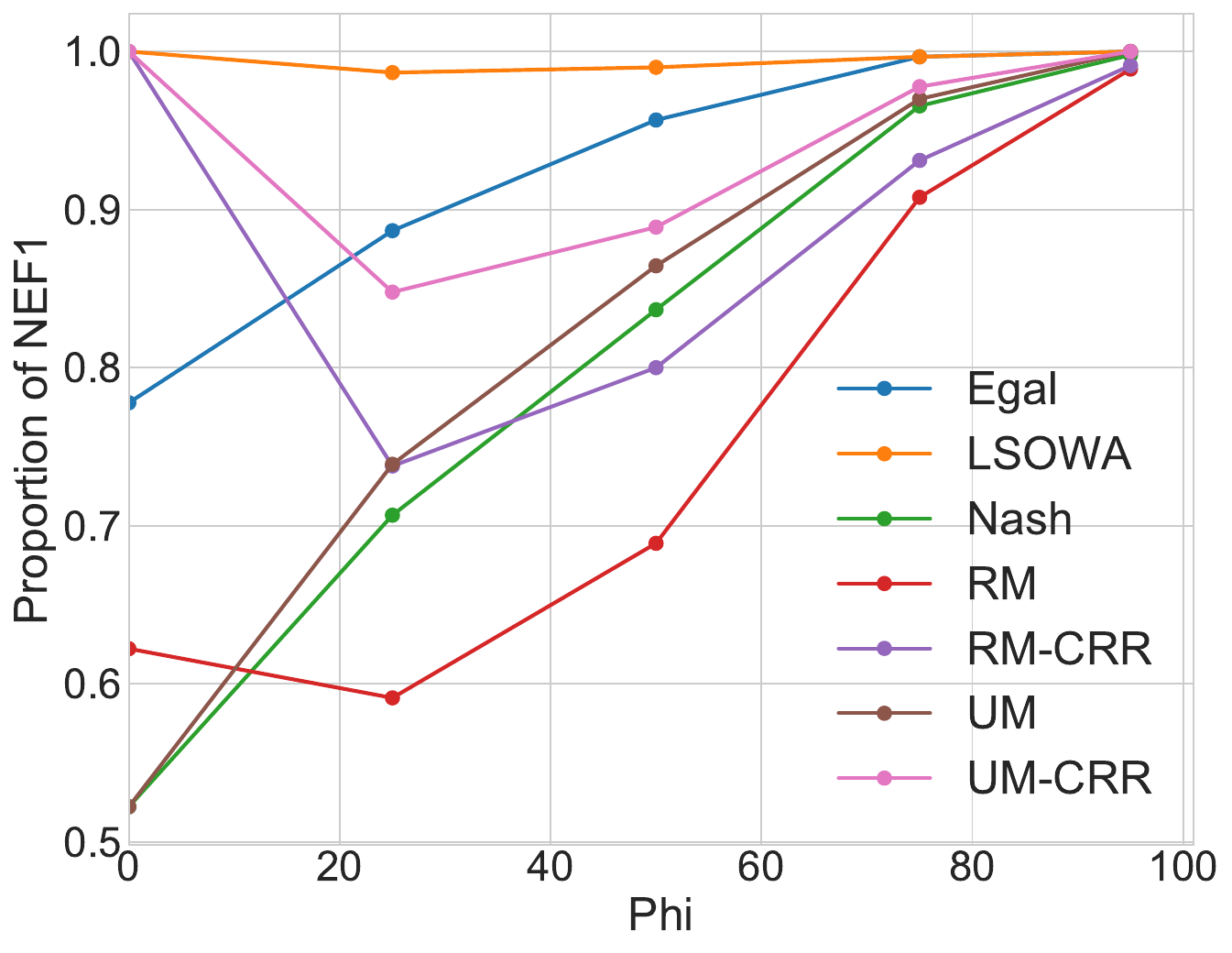}
 \caption{Mean performance of the seven objective functions as we sweep $\phi$ for NEF1.}
 \label{fig:stepping_NEF1}
 \end{figure}

Figure \ref{fig:stepping_NEF1} shows the mean proportion of NEF1 relations that are satisfied for each setting of $\phi$ and provides a good visualization of the effect of correlation on the various algorithms.  We see here that UM-CRR and RM-CRR are able to handle perfectly correlated preferences and provide a solution that has less NEF1 relations than the Egal objective, which is NP-hard to compute.  We also see that the CRR refinement outperforms the UM or RM objective alone for all settings to $\phi$.  Interesting in this graph is the good performance of the $\sum$-OWA objective with a linear OWA (LSOWA).  Similar trends exists for other measures NEF and EF.

\begin{table}[h!]
\resizebox{.95\columnwidth}{!}{
\centering
\begin{tabular}{lccccccc}  
\toprule
& \multicolumn{7}{c}{EF Relations} \\
\cmidrule(r){2-8}
& UM & UM-CRR & RM & RM-CRR & Nash & Linear-$\sum$OWA & Egal  \\
\midrule
MD-2-01		&   0.955  & 0.975  & 0.965  & 0.975  & 0.922  & 0.657  & 0.934 \\
MD-2-02		&   1.0  & 1.0  & 1.0  & 1.0  & 1.0  & 0.835  & 0.998 \\
MD-2-03		&   0.725  & 0.709  & 0.717  & 0.709  & 0.724  & 0.85  & 0.819 \\
\end{tabular}
}

\resizebox{.95\columnwidth}{!}{
\centering
\begin{tabular}{lccccccc}  
\toprule
& \multicolumn{7}{c}{EF1 Relations} \\
\cmidrule(r){2-8}
& UM & UM-CRR & RM & RM-CRR & Nash & Linear-$\sum$OWA & Egal  \\
\midrule
MD-2-01		&   1.0  & 1.0  & 1.0  & 1.0  & 1.0  & 0.823  & 1.0 \\
MD-2-02		&   1.0  & 1.0  & 1.0  & 1.0  & 1.0  & 0.944  & 1.0 \\
MD-2-03		&   0.815  & 0.919  & 0.812  & 0.919  & 0.813  & 0.97  & 1.0 \\
\end{tabular}
}

\resizebox{.95\columnwidth}{!}{
\centering
\begin{tabular}{lccccccc}  
\toprule
& \multicolumn{7}{c}{NEF Relations} \\
\cmidrule(r){2-8}
& UM & UM-CRR & RM & RM-CRR & Nash & Linear-$\sum$OWA & Egal  \\
\midrule
MD-2-01		&   0.951  & 0.966  & 0.961  & 0.965  & 0.902  & 0.615  & 0.927 \\
MD-2-02		&   1.0  & 1.0  & 1.0  & 1.0  & 1.0  & 0.723  & 0.998 \\
MD-2-03		&   0.723  & 0.702  & 0.716  & 0.702  & 0.721  & 0.849  & 0.816 \\
\end{tabular}
}

\resizebox{.95\columnwidth}{!}{
\centering
\begin{tabular}{lccccccc}  
\toprule
& \multicolumn{7}{c}{NEF1 Relations} \\
\cmidrule(r){2-8}
& UM & UM-CRR & RM & RM-CRR & Nash & Linear-$\sum$OWA & Egal  \\
\midrule
MD-2-01		&   0.995  & 1.0  & 0.998  & 1.0  & 0.961  & 0.806  & 0.989 \\
MD-2-02		&   1.0  & 1.0  & 1.0  & 1.0  & 1.0  & 0.918  & 1.0 \\
MD-2-03		&   0.814  & 0.918  & 0.81  & 0.918  & 0.811  & 0.97  & 1.0 \\
\end{tabular}
}

\caption{Results of our CRR algorithms compared to various baseline algorithms for the PrefLib datasets.
}
\label{tab:preflib}
\end{table}

\subsection{Real World Data}

To measure the effectiveness of our algorithms on real-world data we use three datasets from PrefLib: MD-2-01 -- 03 are from small computer science conferences with between 50--175 papers and bids from 25--150 participants.
For all of these datasets, we set the capacity constraints so that each item (paper) is assigned to at least 3 and at most 4 agents.  Each agent gets between 4 and 7 papers.


The results of CRR compared to other objectives are shown in Table \ref{tab:preflib}.  The results are mixed and suggest that there is a rich interplay between the exact bids received and the outcomes of the various algorithms.  Looking first at the EF and EF1 relationships we see that the results are mixed, for some settings, envy goes up when we run the CRR algorithm while in others it actually decreases.  However, for the EF1 relations the CRR refinement is always better and significantly so.  Looking at the NEF and NEF1 relations we see that the CRR algorithm delivers strictly better results in terms of envy than the UM and RM objectives.  
It is interesting to note that in some cases even using the Egalitarian objective does \emph{worse} than the corresponding CRR algorithm.  However, in general, we see that using an Egalitarian objective, though it is NP-hard to optimize for, is a good way to reduce envy on average.

\section{Conclusions}

We have presented a flexible algorithm called CRR that can be parametrized with respect to the target welfare and it can be used in conjunction with any feasibility constraints. 
Since each system designer may have a different sweet spot between welfare and fairness concerns, the CRR algorithm can cater to different `tastes.' For any given welfare requirement, our approach can be seen as providing a fairness serum to the allocation process. 
We  note that if agents have unequal shares or different capacities, this can handled by the algorithm by modifying $N^{Min}$ from $N^{Min}=\{i\in N\mid \not\exists j\in N^{Active} \text{ s.t. } |p(i)|>|p(j)|\}$ to 
$N^{Min}=\{i\in N\midd  |p(i)|<\frac{s_i}{\sum_{j\in N}s_j}|\cup_{j\in N}p(j)|\}.$
We conclude by noting that revisiting standard settings with the goal to incorporate real-life distribution constraints is an important research direction in multi-agent resource allocation. 

%
%

%
%
%
%
%
%
                
      

\clearpage
 \appendix
 

 \section{A. Maximum Nash Welfare Does Not Imply EF1 for Balanced Allocations: Example}

 	There are 4 items, 2 agents, and each agent must get 2 items. The valuations are:
 	\begin{itemize}
 		\item Alice: 5, 5, 2, 2
 	\item Bob:  7, 7, 0, 0
 	\end{itemize}
	
 	There are three possible allocations:
 	\begin{itemize}
 		\item  Alice gets 5,5 and Bob gets 0,0 -- the Nash welfare is 0.
 	\item Alice gets 5,2 and Bob gets 7,0 -- the Nash welfare is 7*7=49.
 	\item Alice gets 2,2 and Bob gets 7,7 -- the Nash welfare is 4*14=56.
 	\end{itemize}
	
 	So the Nash-optimal allocation gives Alice a value of 2+2=4, but it is not EF1 for Alice.
 Note that without the capacity constraints, the Nash-optimal allocation gives Alice 5+2+2=9 and Bob 7, the Nash welfare is 63 and it is indeed EF1.

\clearpage
  \section{B. Stepping Graphs}
  \begin{figure*}[ht!]
  \centering
  \begin{subfigure}{0.4\linewidth}
  	\centering
  	\includegraphics[width=\linewidth,page=1]{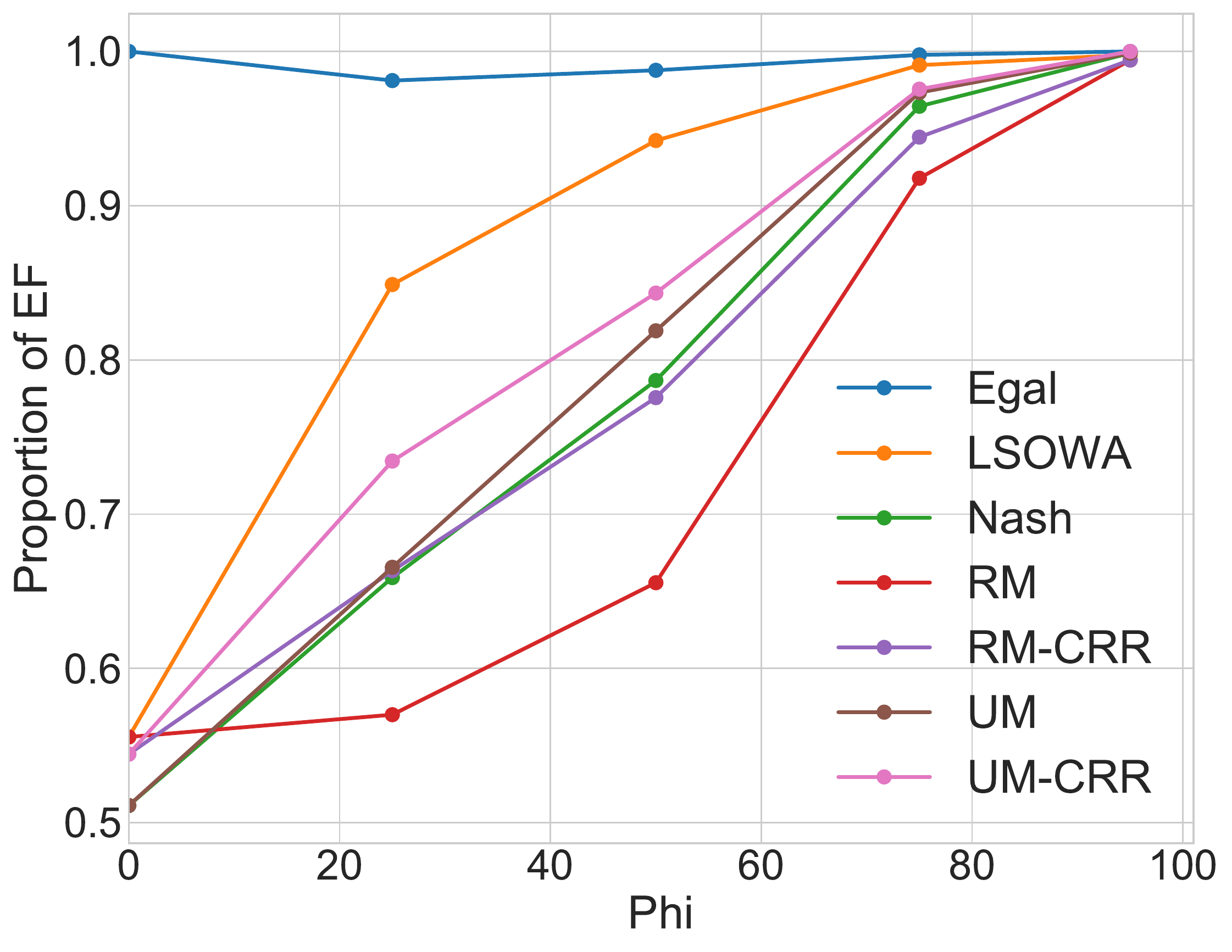}
  \end{subfigure}
  \hspace{2cm}
  \begin{subfigure}{0.4\linewidth}
  	\centering
  	\includegraphics[width=\linewidth,page=2]{./figures/stepping_plots}
  \end{subfigure}
  \hfill
  \begin{subfigure}{0.4\linewidth}
  	\centering
  	\includegraphics[width=\linewidth,page=3]{./figures/stepping_plots}
  \end{subfigure}
  \hspace{2cm}
  \begin{subfigure}{0.4\linewidth}
  	\centering
  	\includegraphics[width=\linewidth,page=4]{./figures/stepping_plots}
  \end{subfigure}
  \caption{Graphs of the performance of the seven objective functions as we sweep $\phi$ for each of our evaluation metrics.}
  \label{fig:stepping_graphs}
  \end{figure*}

  Figure \ref{fig:stepping_graphs} captures the proportions of the various relations as we sweep $\phi$ for each of the evaluation criteria.  These plots can be thought of as an aggregate view of the individual boxplots found in the next section.  Looking at this set of graphs we see that the CRR refinements are outperforming the UM and RM algorithms across the sweep of correlated preferences.  This difference is particularly noteworthy for the NEF1 and NEF plots.

\clearpage
  \section{C. Individual Boxplots}

 Figures \ref{fig:phi_0} through Figure \ref{fig:phi_95} show the distribution of the results for each of the 25 trials for each of the settings to $\phi \in \{ 0.0, 0.25, 0.50, 0.75, 0.95\}$.  These figures give a more nuanced view of what is happening in the interplay between preference correlation and the fairness of the assignments that are found.  Most interestingly is that, as seen in Figure \ref{fig:stepping_graphs} we see the CRR refinements beating out the underling UM and RM allocations at each of the steps.  Another interesting aspect is that when preferences are random ($\phi \in 0.95$) all algorithms are particularly good at finding fair allocations.  The real interesting areas to focus on are when preferences are correlated but not perfectly so.  In these cases, e.g., Figure \ref{fig:phi_0} that NEF is indeed a very stringent requirement and none of the algorithms can guarantee more than 50\% of the relations satisfy the measure.  However, NEF1 is interesting in that many of the algorithms, and especially the CRR algorithms, are able to find an allocation satisfying NEF1 for over 80\% of the pairs.

  \begin{figure*}[hb!]
  \centering
  \begin{subfigure}{0.4\linewidth}
  	\centering
  	\includegraphics[width=\linewidth,page=1]{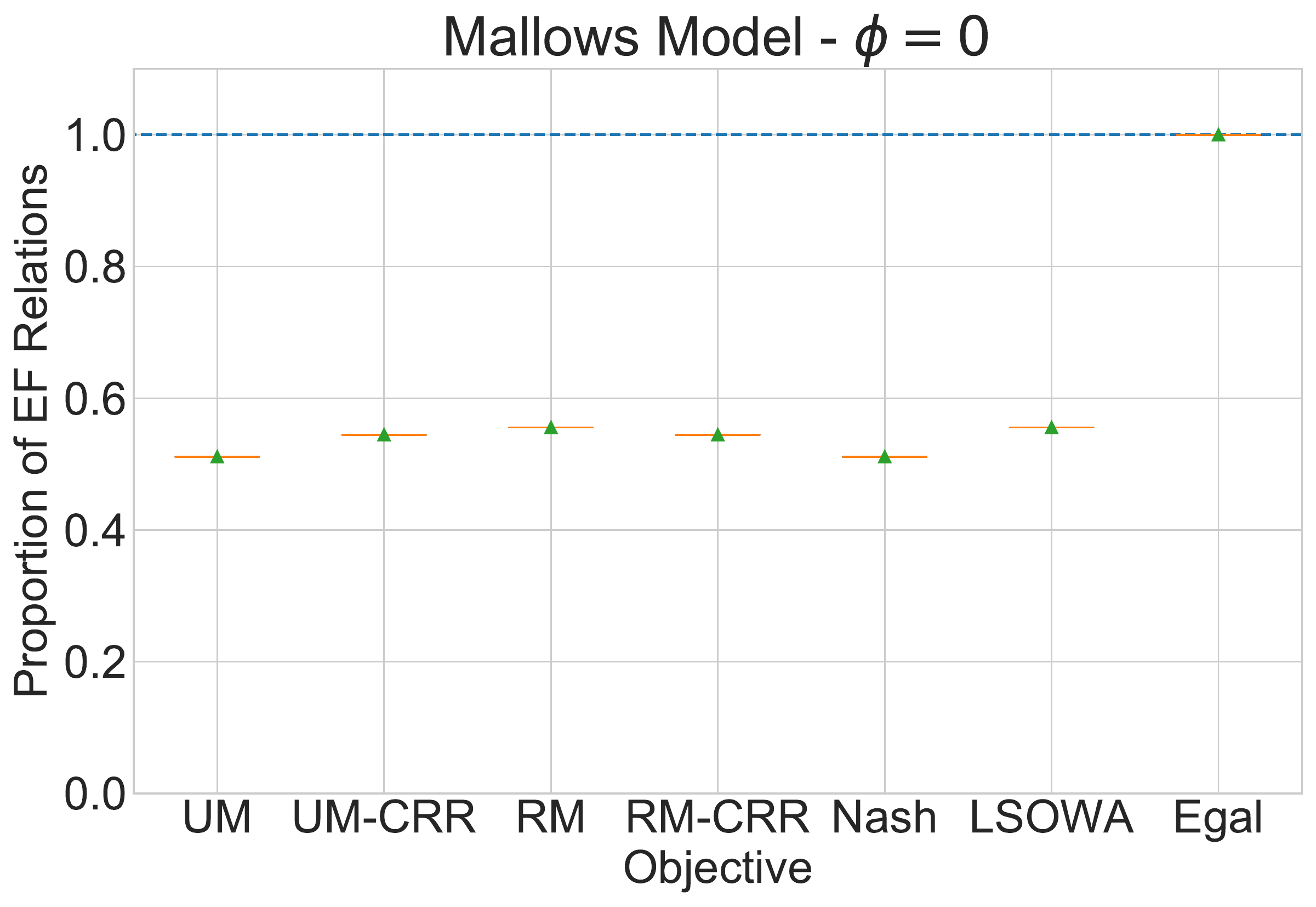}
  \end{subfigure}
  \hspace{2cm}
  \begin{subfigure}{0.4\linewidth}
  	\centering
  	\includegraphics[width=\linewidth,page=2]{./figures/boxplots0}
  \end{subfigure}
  \hfill
  \begin{subfigure}{0.4\linewidth}
  	\centering
  	\includegraphics[width=\linewidth,page=3]{./figures/boxplots0}
  \end{subfigure}
  \hspace{2cm}
  \begin{subfigure}{0.4\linewidth}
  	\centering
  	\includegraphics[width=\linewidth,page=4]{./figures/boxplots0}
  \end{subfigure}
  \caption{Boxplots for $\phi=0$.}
  \label{fig:phi_0}
  \end{figure*}

  \begin{figure*}[ht]
  \centering
  \begin{subfigure}{0.4\linewidth}
  	\centering
  	\includegraphics[width=\linewidth,page=1]{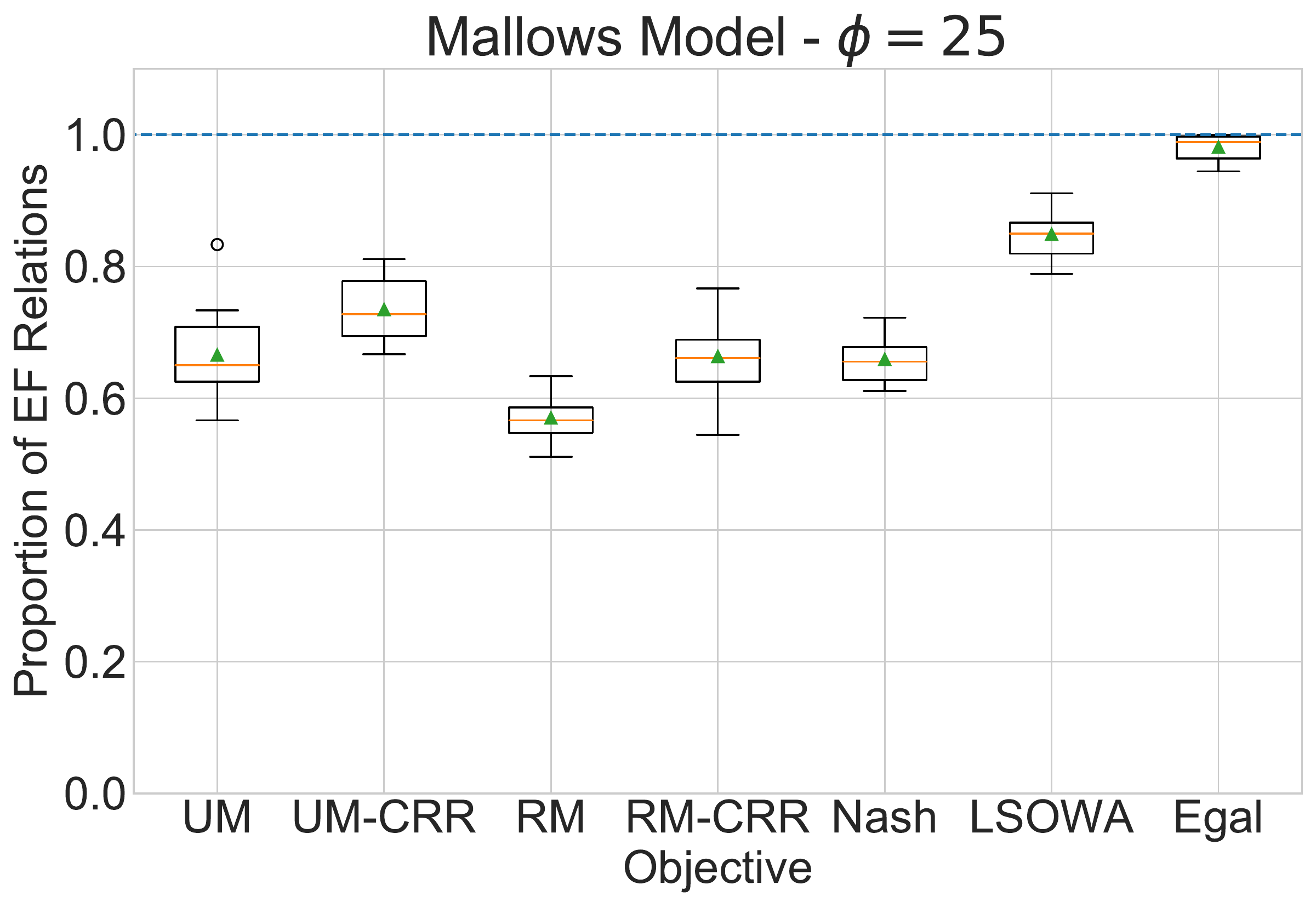}
  \end{subfigure}
  \hspace{2cm}
  \begin{subfigure}{0.4\linewidth}
  	\centering
  	\includegraphics[width=\linewidth,page=2]{./figures/boxplots25}
  \end{subfigure}
  \hfill
  \begin{subfigure}{0.4\linewidth}
  	\centering
  	\includegraphics[width=\linewidth,page=3]{./figures/boxplots25}
  \end{subfigure}
  \hspace{2cm}
  \begin{subfigure}{0.4\linewidth}
  	\centering
  	\includegraphics[width=\linewidth,page=4]{./figures/boxplots25}
  \end{subfigure}
  \caption{Boxplots for $\phi=25$.}
   \label{fig:phi_25}
  \end{figure*}

  \begin{figure*}[ht]
  \centering
  \begin{subfigure}{0.4\linewidth}
  	\centering
  	\includegraphics[width=\linewidth,page=1]{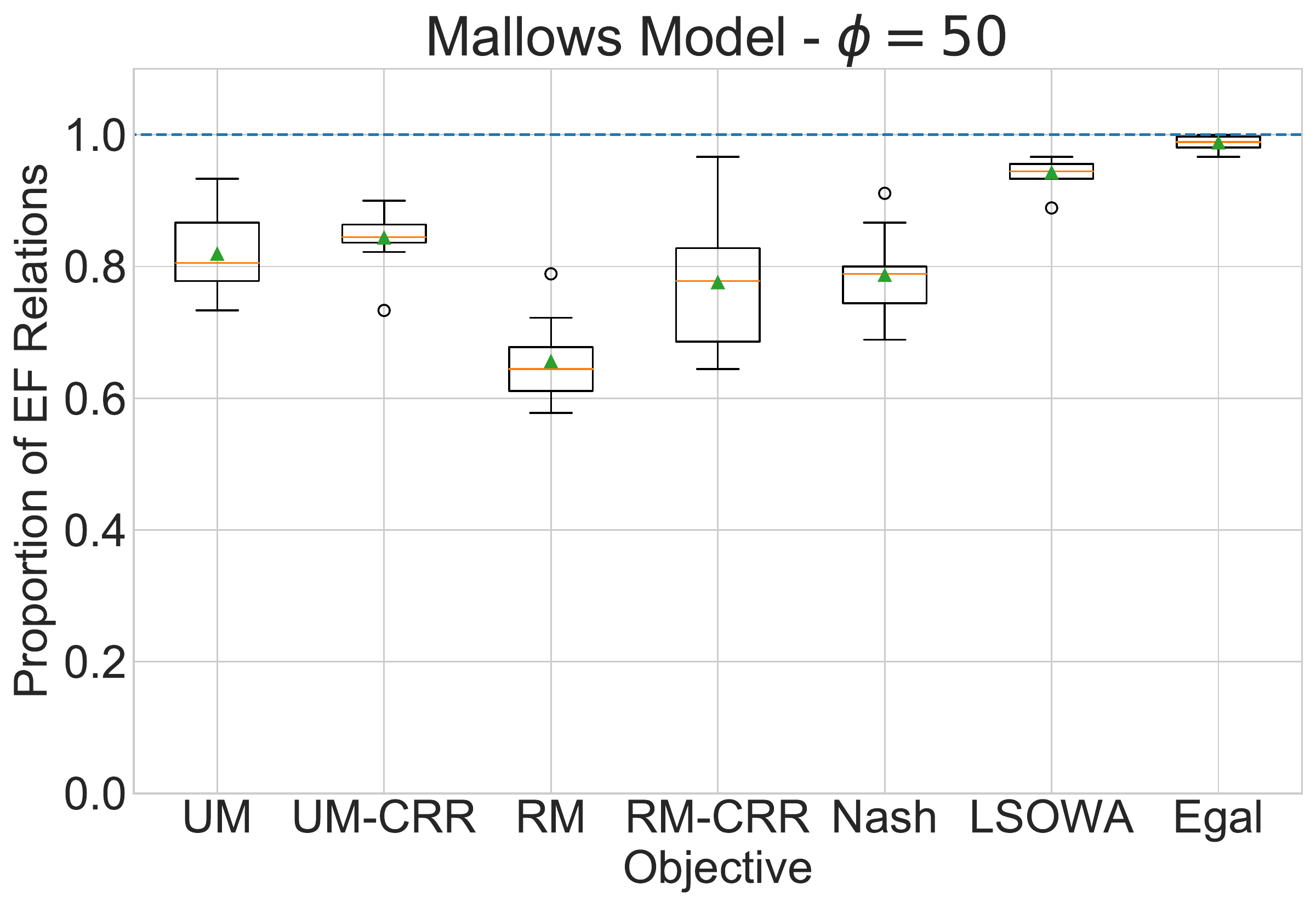}
  \end{subfigure}
  \hspace{2cm}
  \begin{subfigure}{0.4\linewidth}
  	\centering
  	\includegraphics[width=\linewidth,page=2]{./figures/boxplots50}
  \end{subfigure}
  \hfill
  \begin{subfigure}{0.4\linewidth}
  	\centering
  	\includegraphics[width=\linewidth,page=3]{./figures/boxplots50}
  \end{subfigure}
  \hspace{2cm}
  \begin{subfigure}{0.4\linewidth}
  	\centering
  	\includegraphics[width=\linewidth,page=4]{./figures/boxplots50}
  \end{subfigure}
  \caption{Boxplots for $\phi=50$.}
   \label{fig:phi_50}
  \end{figure*}

  \begin{figure*}[ht]
  \centering
  \begin{subfigure}{0.4\linewidth}
  	\centering
  	\includegraphics[width=\linewidth,page=1]{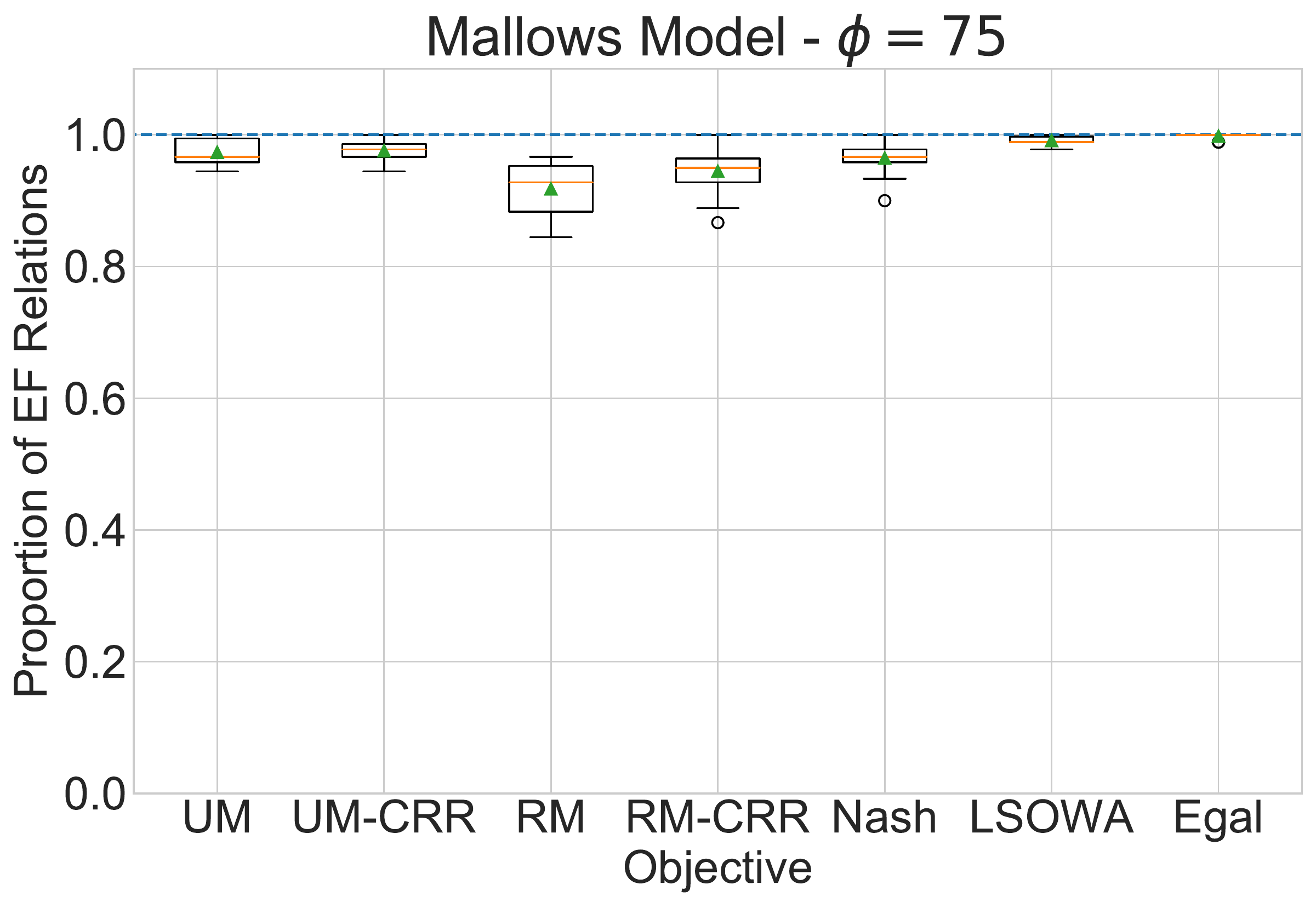}
  \end{subfigure}
  \hspace{2cm}
  \begin{subfigure}{0.4\linewidth}
  	\centering
  	\includegraphics[width=\linewidth,page=2]{./figures/boxplots75}
  \end{subfigure}
  \hfill
  \begin{subfigure}{0.4\linewidth}
  	\centering
  	\includegraphics[width=\linewidth,page=3]{./figures/boxplots75}
  \end{subfigure}
  \hspace{2cm}
  \begin{subfigure}{0.4\linewidth}
  	\centering
  	\includegraphics[width=\linewidth,page=4]{./figures/boxplots75}
  \end{subfigure}
  \caption{Boxplots for $\phi=75$.}
   \label{fig:phi_75}
  \end{figure*}

  \begin{figure*}[ht]
  \centering
  \begin{subfigure}{0.4\linewidth}
  	\centering
  	\includegraphics[width=\linewidth,page=1]{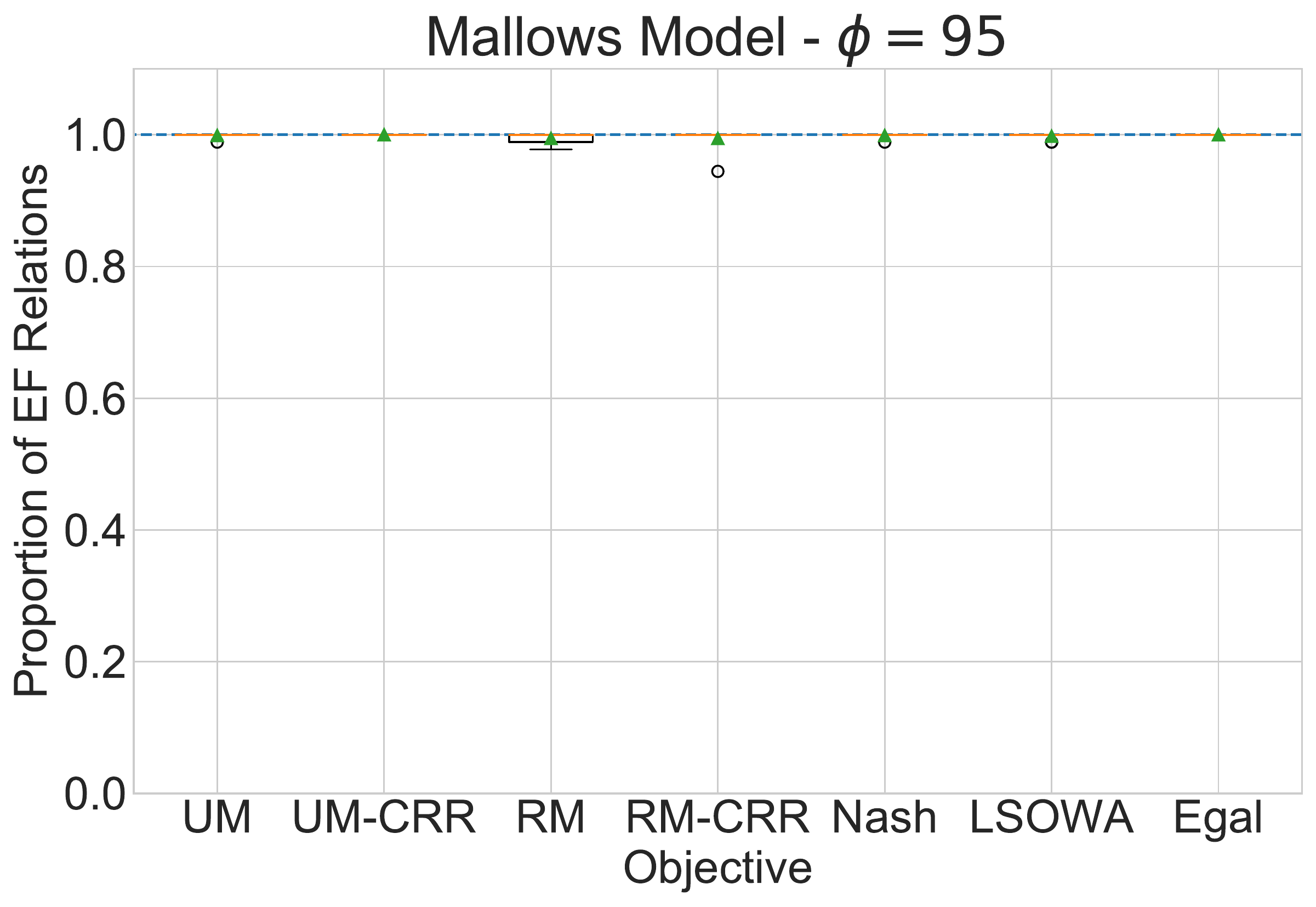}
  \end{subfigure}
  \hspace{2cm}
  \begin{subfigure}{0.4\linewidth}
  	\centering
  	\includegraphics[width=\linewidth,page=2]{./figures/boxplots95}
  \end{subfigure}
  \hfill
  \begin{subfigure}{0.4\linewidth}
  	\centering
  	\includegraphics[width=\linewidth,page=3]{./figures/boxplots95}
  \end{subfigure}
  \hspace{2cm}
  \begin{subfigure}{0.4\linewidth}
  	\centering
  	\includegraphics[width=\linewidth,page=4]{./figures/boxplots95}
  \end{subfigure}
  \caption{Boxplots for $\phi=0.95$.}
   \label{fig:phi_95}
  \end{figure*}

\end{document}